\documentclass{llncs}

\usepackage{amsmath}
\usepackage{amssymb}
\usepackage{amsfonts}

\newcommand{\QED}{$\square$ \smallskip}
\newcommand{\Proof}{\noindent\textit{Proof Sketch. }}

\newcommand{\bt}{\mathcal{T}}
\newcommand{\lldots}{,\ldots ,}
\newcommand{\nil}{\langle \ \rangle}
\newcommand{\wwedge}{\ \wedge \ }
\newcommand{\bSearch}{{\tt bSearch}}
\newcommand{\Rec}{{\tt Rec}}
\newcommand{\rev}{\text{rev}}
\newcommand{\Cond}{{\tt Cond}}
\newcommand{\error}{{\tt fault}}

\newcommand{\lst}[1]{\langle \ #1 \ \rangle}
\newcommand{\conc}{{\tt conc}}
\newcommand{\cons}{{\tt cons}}
\newcommand{\head}{{\tt head}}
\newcommand{\tail}{{\tt tail}}

\newcommand{\EXP}[2]{{{\tt{#1}}{\tt{exp}}({#2})}}

\renewcommand{\c}{{\mathfrak c}}
\newcommand{\q}{{\tt q}}
\newcommand{\qo}{{\tt q_0}}
\renewcommand{\b}{{\tt b}}

\begin{document}
\title{On Termination of Transactions over Semantic\newline Document Models\thanks{The research was supported by the Russian Science Foundation (Grant
No. 17-11-01176)}}

\author{Andrei Mantsivoda\inst{1,3,4} \and Denis Ponomaryov\inst{2,3,4}}

\institute{Irkutsk State University \and Ershov Institute of Informatics Systems \and Sobolev Institute of Mathematics\newline \and Novosibirsk State University}

\maketitle

\begin{abstract}
We consider the framework of Document Modeling, which lays the formal basis for representing the document lifecycle in Business Process Management systems. We formulate document models in the scope of the logic-based Semantic Modeling language and study the question whether transactions given by a document model terminate on any input. We show that in general this problem is undecidable and formulate sufficient conditions, which guarantee decidability and polynomial boundedness of effects of transactions.    
\end{abstract}


\section{Introduction}\label{Sect:Intro}
In \cite{docmodels} a Document Modeling approach has been proposed as a fundamental basis for document processing in Business Process Management Systems (BPMS). Importantly, within this approach basic entities and primitives have been identified, which are common to BPMS such as Enterprise Resource Planning Systems, Customer Relationship Management Systems, etc. The approach rests on the natural idea that document lifecycle lies at the core of these systems. Typically, there is a static part, which describes the forms and statuses of documents (i.e., a schema), and a dynamic part, which describes changes in documents (i.e., transactions over them). In contrast to conventional architectures of BPMS, the approach of the Document Modeling shows that both parts can be given in a fully declarative fashion, thus making programming unnecessary. It suffices to describe the static part of a document model by giving a specification to document forms and fields, and to describe the dynamic part by defining transactions, their conditions, and effects. Then, given an initial state of a document model (a collection of documents), the natural problem is to compute a state (an updated collection of documents), which results from the execution of a sequence of transactions. It is argued within the Document Modeling approach that this problem can be solved with the tools of formal logic such as automated inference or model checking.

In \cite{MantsivodaPonomaryov}, the ideas of the Document Modeling have been implemented in a logical framework in terms of the language of the Semantic Programming (aka Semantic Modeling) \cite{bib41}. It has been shown that the approach of the Document Modeling implemented this way goes beyond the common capabilities of today's Business Process Management Systems. In particular, it allows for checking document models for consistency and solving important problems like projection (e.g., what documents will be created after an accountant performs certain actions) and planning (e.g., what actions must be made in order to get an item on stock). The method follows the same line with some of the well-known approaches like Situation Calculus \cite{ReiterBook2} and similar formalisms, but it addresses the topic of Business Process Management, which is a novel area of application for logic-based formalisms.

Obviously, an important question is how hard the above mentioned problems are from the computational point of view. In this respect, the key problem is computing effects of transactions over a document model. Transactions can be fired due to an input of an oracle (a user or an algorithm, which provides some input to a document model), which in turn, can cause other transactions to fire, and so on. Thus potentially, this can result in an infinite chain of updates of a document model, under which a finite resulting state is never obtained. We consider this problem in the paper and formulate a number of complexity results, which demonstrate the expressiveness of document models. 

The contributions of this work are as follows. We refine the formalization of the Document Modeling given in \cite{MantsivodaPonomaryov} and provide a more succinct formalization in an extension of the language of the Semantic Modelling with (non-standard) looping terms. We formulate the problem of transaction termination over document models and show that in general it is undecidable. 
Then we describe a sufficient condition, which guarantees decidability. For this we introduce a formal definition of a locally simple document theory (the notion previously discussed in \cite{docmodels2}) and we show that over any such theory transaction termination is decidable. Then we estimate the complexity of computing effects of transactions and identify a case when they are polynomially bounded. 

\section{Preliminaries}\label{Sect:Preliminaries}
Document Modeling follows the idea of declarative representation of documents and transactions over them. A document model consists of a description of fields, which can appear in documents (their cardinality and default values), a definition of document forms (given as collections of fields), and a definition of so called \textit{daemons}, which specify conditions and effects of transactions and field triggers. Transactions can be fired on an input of a user or an external procedure (e.g., a Machine Learning algorithm like in \cite{Vityaev}), or they can be fired by other transactions. Field triggers can be viewed as a special kind of transactions, but they can fire only in the event of changing a value of some document field. 

The formalism of the Document Modeling includes at least three ingredients that can influence the complexity of computation. The first one is the set of operators over field values. In real-world applications of the Document Modeling, the language is restricted to basic arithmetic operations (like, summation, subtraction, etc.), which can be computed efficiently. For this reason, we do not consider the whole variety of operators over field values in the paper. We describe only basic operations and examples of their implementation in order to show that they make no contribution to the complexity of computing effects of transactions. The second ingredient is the query language used in the Document Modelling to describe collections of documents, which have certain properties. Transactions can refer to document collections given by queries and hence, the complexity of the query language influences the complexity of computing effects of transactions. We leave this effect out of the scope of this paper and focus on the complexity of transactions caused solely by their relationships to each other. For this, we adopt a simple query language implemented by predefined document filters, which can be used in the definition of transactions and are computationally simple. In the remaining part of this section, we introduce basics of the Semantic Modeling and  conventions used in this paper. We refer an interested reader to \cite{bib41}-\cite{bib32} for details on the Semantic Modeling. 

\subsection{Basics of the Semantic Modeling}\label{Sect:BasicsSemModeling}
The language of the Semantic Modeling is a first-order language with sorts `urelement' and `list' , in which only bounded quantification of the following form is allowed:
\begin{itemize}
\item a restriction onto the list elements $\forall x\in t$ and $\exists x\in t$;
\item a restriction onto the initial segments of lists $\forall x\sqsubseteq t$ and $\exists x\sqsubseteq t$.
\end{itemize}
where $t$ is a list term. A list term is defined inductively via constant lists, variables of sort `list', and list functions given below.
A \textit{constant list} (which can be nested) is built over constants of sort `urelement' and a constant $\nil$ of sort `list', which represents the empty list. The \textit{list functions} are:
\begin{itemize}
\item $\head$ -- the last element of a non-empty list and $\nil$, otherwise; 

\item $\tail$ -- the list without the last element, for a non-empty list, \newline\phantom{aaaaaa} and $\nil$, otherwise; 

\item $\cons$ -- the list obtained by adding a new last element to a list; 

\item $\conc$ -- concatenation of two lists; 
\end{itemize}
\smallskip

Terms of sort `urelement' are standard first-order terms. The predicates $\in, \sqsubseteq$ are allowed to appear in $\Delta_0$-formulas without any restrictions, i.e., they can be used in bounded quantifiers and atomic formulas. 

\medskip

Formulas in the language above are interpreted over hereditarily finite \textit{list superstructures} $HW(\mathcal{M})$, where $\mathcal{M}$ is a structure. Urelements are interpreted as distinct elements of the domain of $\mathcal{M}$ and lists are interpreted as lists over urelements and the distinguished `empty list' $\nil$. In particular, the following equations hold in every $HW(\mathcal{M})$ (the free variables below are assumed to be universally quantified):
{\footnotesize
\begin{equation*}
\begin{aligned}
& \neg \exists x \ x\in \nil \\
& \cons(x,y) = \cons(x',y') \rightarrow x=x' \wedge y=y' \\
& \tail(\cons(x,y)) = x, \ \ \head(\cons(x,y)) = y \\
& \tail(\nil) = \nil, \ \ \head(\nil) = \nil \\
& \conc(\nil, x) = \conc(x, \nil) = x \\
& \cons(\conc(x,y), z) = \conc(x, \cons(y,z)) \\
& \conc(\conc(x,y), z) = \conc(x, \conc(y,z))
\end{aligned}
\end{equation*}
}
It was shown in \cite{OspichevPonomaryov} that for any appropriate structure $\mathcal{M}$, there exists a representation of its superstructure of finite lists $HW(\mathcal{M})$, in which the value of any variable-free list term $t$ can be computed in time polynomial in the size of $t$ (given as as string). Throughout the text,  we omit subtleties related to the representation of hereditarily finite structures and we assume that for any variable-free list term $t$ one can compute a constant list $t'$ in time polynomial in the size of $t$ such that $HW(\mathcal{M})\models t = t'$, for any structure $HW(\mathcal{M})$ under consideration. For list terms $t_1\lldots t_n$, $n\geqslant 1$, we will use $\lst{t_1\lldots t_n}$ as a shortcut for the term $\cons(\cons(\cons(\nil, t_1), t_2)\ldots , t_n)\ldots )$. For a list $s$, the notation $|s|$ stands for the number of elements in $s$.

In \cite{condterms,Doklady,Recursive}, the basic language of the Semantic Modeling was extended with non-standard list terms, which represent conditional operators (they correspond to the common `if-then-else' or `switch' constructs of programming languages), bounded list search, and bounded recursion (similar to the restricted `while' operator). We refer to the obtained language as $\mathcal{L}$. The non-standard terms in $\mathcal{L}$ are called $\Cond$-, $\bSearch$- and $\Rec$-terms, respectively, and are defined as follows. By default any standard term in the language of the Semantic Modeling is a $\mathcal{L}$-term and any formula of the language of the Semantic Modeling is a $\mathcal{L}$-formula.

If $t$ and $\theta(\overline{v},x)$ is a $\mathcal{L}$-term of sort list and $\mathcal{L}$-formula, respectively, then the expression $\bSearch(\theta, t)(\overline{v})$ is a \textit{$\bSearch$-term}. It is equal to the last element $a$ of $t(\overline{v})$ such that $\theta(\overline{v}, a)$ holds and it is equal to $t(\overline{v})$, otherwise (i.e., if there is no such $a$). 

If $\theta_0\lldots \theta_n$ are $\mathcal{L}$-formulas and $q_1\lldots q_{n+1}$ are $\mathcal{L}$-terms, where $n\geqslant 0$, then the term
$\Cond[\theta_1, q_1][\theta_2, q_2]\ldots[\theta_n, q_n][q_{n+1}](\overline{v})$ is a \textit{$\Cond$-term} term with the following interpretation:
{\footnotesize 
\begin{center}
$t(\overline{v})=\begin{cases} q_1(\overline{v}) &\mbox{if
} \theta_1(\overline{v}) \\ q_2(\overline{v}) &\mbox{if }
\theta_2(\overline{v})\land\neg\theta_1(\overline{v}) \\ \ldots \\
 q_{n}(\overline{v}) &\mbox{if }
\theta_n(\overline{v})\land\neg\theta_1(\overline{v})\land\neg\theta_2(\overline{v})\land\ldots\land\neg\theta_{n-1}(\overline{v})
\\q_{n+1}(\overline{v}) &\mbox{if } \neg\theta_1(\overline{v})\land\neg\theta_2(\overline{v})\land\ldots\land\neg\theta_{n}(\overline{v})
\end{cases}$
\end{center}
}
Finally, if $f(\overline{v}), h(\overline{v},y,z)$ and $t(\overline{v})$ are $\mathcal{L}$-terms of sort list then the expression $\Rec[f,$ $h,t](\overline{v})$ is a \textit{$\Rec$-term} and its value is given by $g(\overline{v},t)$ with the following definition:
\begin{itemize}
	\item $g(\overline{v},\nil)=f(\overline{v})$
	\item $g(\overline{v},cons(\alpha,b))\!=\!h(\overline{v},g(\alpha),b)$, for any lists $\alpha,b$ such that $cons(\alpha,b)\sqsubseteq t$
\end{itemize}	

\smallskip

In this paper, we refine the formalization of the Document Modeling from \cite{MantsivodaPonomaryov} in the language of the Semantic Modeling extended with the above mentioned non-standard list terms. In particular, we obtain a more succinct formalization in comparison with \cite{MantsivodaPonomaryov}. Further in Section \ref{Sect:DocumentTheories}, we will introduce document theories, which formalize the key ingredients of the Document Modeling approach, and in the next section we describe conventions used in our formalization.

\subsection{Conventions in Formalization of Document Theories}
We use the following notions and informal conventions:

\begin{itemize}
\item There are pairwise disjoint finite sets $FieldNames$, $FormNames$,\newline $FilterNames$, and $TransNames$ of constants of sort \textit{urelement}, which provide document field, form, and filter names, and transaction names, respectively, which can be used in the axioms of a document theory. \smallskip


\item Natural numbers are modelled in a straightforward way as lists consisting of $n$ empty lists, for $n\geqslant 0$, and $0$ is represented by the empty list $\nil$ (we also show how to model real numbers in a decimal representation with a given precision). \smallskip

\item An \emph{instruction} is given as a list of the form $\lst{formName, CreateDoc}$ (in which case it is called \textit{CreateDoc-instruction}) or $\lst{value, fieldName,$ $docID, SetField}$ (a \textit{SetField-instruction}), or $\lst{params, docID,$ $transName}$ (a \textit{transaction}),  where $formName\in$ $FormNames$, $fieldName$ $\in Field\-Names$, $transName\in TransNames$, $docID$ represents a natural number, and $value$, $params$ are some lists, which specify a field value and transaction parameters, respectively. \smallskip

\item A \emph{queue} is a list of instructions to be executed. A queue is updated by daemons, which implement actions on the events such as changing a field value in a given document or executing a transaction. Creating a new document triggers no events. \smallskip

\item A \emph{situation} is a list of instructions, which represents the history of executed instructions. The last executed instruction appears first in a situation.  \smallskip
 
\item A \emph{field} is given as a list, with the head being an element of $FieldNames$ and the tail being a list, which represents a value for a field. Every field has a default value it gets when a new document is created.  \smallskip

\item A \emph{document} is a list of fields (the order of fields in the list is arbitrary). \smallskip

\item A (document) \emph{model} is a list consisting of tuples $\lst{sit, form, doc, ID}$, where $ID$ corresponds to a natural number, $doc$ is a document, $form\in FormNames$, and $sit$ a situation. A model stores a version of each document in each situation which has ever taken place. The head of this list is a tuple, in which the situation is the current one, i.e., it consists of instructions (a history) that have given the model. 
\end{itemize}
\smallskip

Situations represent contexts, in which documents are created or modified, and this information can be used in querying a document model. We note that this feature is irrelevant for the results in this paper, but we prefer to keep situations to comply with the original formalization of document models from \cite{MantsivodaPonomaryov}.

A \textit{document theory} consists of axioms, which specify document fields, forms, filters (i.e., the  static structure of documents and query templates), and axioms for the dynamic part. The latter is given by so called \emph{daemons} (similar to the notion used in process programming), which specify the instructions that must be executed whenever certain event happens (i.e., whenever a value of a specific field in a document is changed or a certain transaction is fired).  
Although formally we distinguish between CreateDoc-, SetField-instructions and transactions, we make no terminological difference between them when talking about the transaction termination problem. The results on computing effects of transactions refer to the instructions of the form above as well.

\section{Document Theories}\label{Sect:DocumentTheories}
We define a document theory $\bt$ as a theory in signature $\Sigma$, where $\Sigma$ consists of the list functions introduced in Section \ref{Sect:BasicsSemModeling} and the predicate and function symbols introduced in the axioms below. In particular, $\Sigma$ contains pairwise disjoint finite subsets of constants $FieldNames$, $FormNames$, $FilterNames$, and $TransNames$, which specify field, form, filter, and transaction names, which can be used in the axioms of $\bt$. The set $FormNames$ is supposed to be non-empty. Besides, $\Sigma$ contains distinguished constants $CreateDoc$ and $SetField$, $ExecTrans$, $\error$, which are used to represent instructions, and \textit{fault} (analogous to exception in programming languages). 

We formulate the axioms of $\bt$ in the language of the Semantic Modeling with non-standard terms. Initially, this language contains only two sorts: \textit{urelement} and \textit{list}. For convenience, we will assume that there is also a subsort $Real$ of the sort \textit{list}, which corresponds to (non-negative) real numbers with a given precision (denoted further as $prec$). In the following subsection, we define the sort $Real$, together with the corresponding predicates and functions, and we show how basic arithmetic operations can be implemented via list terms. In general, there are many such implementations possible, so the next subsection is best viewed as a number of introductory examples to the language of the Semantic Modeling. The only important observation is that the proposed implementation is tractable, as stated by Lemma \ref{Lem:TractabilityValueTerms} in Section \ref{Sect:DocumentValueTerms}. Throughout the text we assume that all the free variables in formulas are universally quantified.

\subsection{Numeric Terms and Predicates}

Let us define $Nat(x) \equiv \forall t \in x \ t = \nil$. For a natural number $n\in\omega$, denote by $\bar{n}$ the list consisting of $n$ empty lists.
Given $prec\in\omega$, $prec\geqslant 1$, we define a subsort $Real$ of the sort \textit{list} as follows:
{\footnotesize $$Real(x)\  \equiv \ len(x)=\overline{prec} \ \wedge \ \forall t\in x \  Nat(t) \wedge len(t)\sqsubseteq \overline{9} $$} where $len(x)$ is an abbreviation for the term $Rec[\nil, \cons(g(\alpha), \nil), x](x)$, i.e., $len(x)$ gives the number of elements in a list $x$.
In other words, we assume that a list of sort $Real$ corresponds to the decimal representation of a real number using $prec$-many digits, for a fixed number $prec\in\omega$.  


For lists $x, i$, let $x.i$ be a shortcut for the term
{\footnotesize $$\Cond[\neg Nat(i) \vee \neg (i\sqsubseteq len(x)), \ \error] \ [\ \Rec[\nil, b, i] \ ]$$}
i.e., it gives the constant list $\error$ if $i$ does not correspond to a natural number or $i$ is greater than the number of elements in $x$. Otherwise it gives the $i$-th element of $x$.

For lists $x,y$, let $x < y$ be the conjunction of $Real(x) \wedge Real(y)$ with
{\footnotesize 
\begin{multline*}
\exists i\sqsubseteq \overline{prec} \ (\ x.i \sqsubseteq y.i \wedge x.i\neq y.i \ \wedge 
\forall j\sqsubseteq \overline{prec} \ (i \sqsubseteq j \rightarrow x.j = y.j) \ )
\end{multline*}
}
i.e., we assume that the first digit of a real number given by a list $x$ is $\head(x)$.
The corresponding predicate $x\leqslant y$ is defined similarly. \smallskip

For a list $t$, let $min(t)$ be a notation for the term
{\footnotesize $$\Cond[t=\nil \vee \exists s\in t (\neg Real(s)), \ \error]\ [\ \Rec[\ \head(t), \ \Cond[b<g(\alpha), b][g(\alpha)], \ t \ ]\ ]$$}
The term $\max(t)$ is defined similarly. \smallskip

Finally, for lists $x,y$, let $x+y$ be a shortcut for the term 
{\footnotesize $$\Cond[\neg(Real(x)\wedge Real(y)) \vee \tail(sum)=\overline{1}, \ \error ]\ [\head(sum)]$$} 
where \phantom{a}
{\footnotesize $sum \equiv \Rec[\lst{\nil,\nil}, \ \cons(\ \tail(s), \cons(\head(g(\alpha)), \head(s)) \ ), \ \overline{prec}],$ }
\begin{center}
{\footnotesize $s \equiv sumnat(\ \conc( \ x.\cons(\alpha,b), \tail(g(\alpha)) \ ), \ y.\cons(\alpha,b) \ ),$ } 
\end{center}
{\footnotesize 
\begin{equation*}
\begin{aligned}
  sumnat(x,y) \equiv \Cond[\overline{10} \leqslant \conc(x,y), \ \cons( \ \overline{1}, mod10(\conc(x,y)) \ ) ][\cons(\overline{0}, \conc(x,y))]
\end{aligned}
\end{equation*}
}
and \phantom{aa} {\footnotesize $mod10(x)  \equiv  \head( \ \Rec[\lst{\nil, \nil},  
\Cond[\tail(g(\alpha))=\overline{10},$}

{\footnotesize $\cons(\tail(g(\alpha)), \cons(\head(g(\alpha)), b))] \ [\cons(\ \cons(\tail(g(\alpha)), b), \head(g(\alpha)) \ ) ], x] \ )$}.

\medskip

We note that negative reals and other arithmetic operations, e.g., subtraction, multiplication, etc., can be defined in a similar fashion. 

\medskip

Let $prec=k+m$, where $k,m$ are some constants, which give the length of the integer/fractional part of real numbers, respectively. For a (non-negative) real number $n$, let $dec(n)$ be the decimal representation of $n$ such that the number of digits in the integer and fractional part of $dec(n)$ is exactly $k$ and $m$, respectively. This is achieved by using auxiliary zeros, e.g., for $n=3/2$ and $k=m=2$, we have $dec(n)=01.50$. If $dec(n)$ exists, let $List(n)$ be the list representation of $dec(n)$, i.e., the list such that $len(List(n))=\overline{prec}$ and for all $i\in\{1\lldots prec\}$ and $j\in\omega$, it holds $List(n).\overline{i}=\overline{j}$ iff $j$ is the $(prec+1-i)$-th digit in $dec(n)$. 

The following lemma sums up the properties of the given formalization:

\begin{lemma}[Implementation of Arithmetic with Precision]\label{Lem:Arithmetic} 
Let $HW(\mathcal{M})$ be a list superstructure and $prec\in\omega$ a precision. For any (non-negative) real numbers $a_i$ such that $dec(a_i)$ exists, for $i=1\lldots n$ and $n\geqslant 3$:
\begin{itemize}
\item $dec(a_1)\!\propto\!dec(a_2)$ iff $HW(\mathcal{M})\models List(a_1)\propto List(a_2)$, for $\propto\in\{<, = \}$
\item $dec(a_1) + dec(a_2) = dec(a_3)$ iff $HW(\mathcal{M})\models List(a_1)+List(a_2)=List(a_3)$ 
\item $dec(dec(a_1) + dec(a_2))$ does not exist iff $HW(\mathcal{M})\models List(a_1)+List(a_2)=\error$
\end{itemize}
For $n\geqslant 1$, the value of $min(\lst{List(a_1)\lldots List(a_n)})$ or $max(\lst{List(a_1),$  $\ldots , List(a_n)})$ in $HW(\mathcal{M})$ is $List(a)$ iff $a$ is minimal/maximal among $dec(a_1),$ $\ldots , dec(a_n)$, respectively.
\end{lemma} 


\subsection{Document Terms}\label{Sect:DocumentValueTerms}
Let us introduce notations for terms, which are used to access documents and field values in a document model.

The following term gives the last used ID for a document in a model:
{\footnotesize 
$$GetLastDocID(model) \equiv max(\cons(\Rec[\nil, \cons(g(\alpha),\head(b)), model]), \ \overline{0})$$
}
i.e., it implements a search for the greatest value occurring as the head of a tuple from $model$ and outputs $\overline{0}$ if there are no documents in the model.

The next term gives the last version of a document (from a model) by its ID. It implements search for the last tuple with a given ID (contained in a model) and outputs the found document. If no tuple with the given ID is present in the model, the term gives \error.
{\footnotesize 
    $$GetDocByID(docID, model) \equiv \Cond[{\tt doctuple}=model, \ \error]    [{\tt doctuple}]$$
}
where {\footnotesize ${\tt doctuple}=\bSearch[\head(x)=docID, \ model]$}. 

The next term provides a field value from the last version of a document with a given ID: 
{\footnotesize 
\begin{multline*}
    GetFieldValue(docID, fieldName, model) \equiv \\ 
    \Cond[{\tt document}=\error, \error]
    [\ \tail(\bSearch[\head(x)=fieldName, {\tt document}])\ ]  
\end{multline*}
}
where {\footnotesize ${\tt document} = \head(\tail(GetDocByID(docID, model)))]$}.
\smallskip

Finally, we define the term $FindFieldPosition$, which ``splits'' a given document into a partitioned one (denoted as $pdocument$ below), which has the form $\lst{list_1, list_2}$ such that $conc(list_1,list_2)=document$ and $head(list_1)$ is a field with the required name (if there exists one in a document). This auxiliary term is employed in the axioms of a document theory to implement change of a field value in an existing document:
{\footnotesize 
\begin{multline*}
    FindFieldPosition(document, fieldName) \equiv \\
    \Cond[\tail({\tt pdocument})=\nil, \error][{\tt pdocument}]
\end{multline*}
}
where 
{\footnotesize 
\begin{multline*}
{\tt pdocument} = \Rec[\ \nil, \\
\Cond[\  \head(\tail(g(\alpha)))=fieldName,
 \lst{\tail(g(\alpha)), \cons(\head(g(\alpha)), b)} \ ] \\
 [\ \lst{\cons(\tail(g(\alpha)), b), \nil} \ ] , \ document \ ]  
\end{multline*}
}

Now we define by induction the notion of \textit{document term}, which generalizes the definitions above.

\begin{definition}[Document Term]
Any standard list term (i.e., which does not contain \Cond-, \bSearch, or \Rec-terms) is a document term. If $s,t,u,i$ are document terms then $s.i$, $s+t$, $min(s)$, $max(s)$, $GetLastDocID(s)$,  $GetDocByID(s,t)$, and $GetFieldValue(s,t,u)$ are document terms. 

The definition of document term is complete.
\end{definition}

An important property is that these terms are computationally tractable as stated in the following lemma.

\begin{lemma}[Tractability of Document Terms]\label{Lem:TractabilityValueTerms}
For any $prec\in\omega$, document terms $s(\overline{u})$, $t(\overline{v})$, and vectors of constant lists $\overline{a}$, $\overline{b}$: 
\begin{itemize}
\item a constant list $c$ such that $HW(\mathcal{M})\models s(\overline{a})=c$, for any list superstructure $HW(\mathcal{M})$ (which contains all the urelements from $s,t,\overline{a},\overline{b}$), can be computed in time polynomial in the size of $s(\overline{a})$ and $\overline{prec}$;
\item it can be decided in time polynomial in the size of $s(\overline{a})$,  $t(\overline{b})$, and $\overline{prec}$ whether $s(\overline{a})\propto t(\overline{b})$, for $\propto\in\{<,=\}$, holds in any structure as above.
\end{itemize}
\end{lemma}
\Proof
The first point of the lemma is proved by induction on the form of the term $s$. For a standard list term, the claim readily follows from Lemma 2 in \cite{OspichevPonomaryov}. For an arbitrary document term $s$ the claim is shown by analyzing the syntactic form of the terms $.i$, $+$, $min()$, $max()$, $GetLastDocID()$,  $GetDocByID()$, and $GetFieldValue()$. It follows from their definition that each of these terms can be computed in polynomial time in the size of their parameters and $\overline{prec}$. The second point of the lemma is shown by an analysis of the definition for $<$: it gives a polynomial time algorithm to verify whether there is a segment $i\sqsubseteq \overline{prec}$, for which the condition from the definition of $<$ is true. \QED

\subsection{Axioms of Document Theory}
A \textit{document theory} has the form $\bt=\bt_{f}\cup\bt_{s}\cup\bt_{d}$, where the theory $\bt_f$ gives predefined filters, which can be used to select collections of documents, $\bt_s$ gives definitions to document fields and forms (i.e., it describes the data schema, hence, the subscript $s$), and $\bt_d$ describes possible transactions and triggers, their execution rules, and instruction processing rules, which generate documents or update existing ones. Thus, $\bt_d$ describes the dynamics of documents (hence, the subscript $d$). 
\smallskip
First, let us introduce auxiliary terms, which will be used in axioms of $\bt$. The first one gives a form name of a document
{\footnotesize $$\text{Form}(document) \equiv \head(\tail(\tail(document)))$$}
while the second one gives a list, in which the order of elements is reversed:
{\footnotesize $$\rev(list) \equiv \Rec[\nil, \ conc(\lst{b}, g(\alpha)), \ list]$$}
We begin with a definition of theory $\bt_f$. 
For each $name\in FilterNames$, it contains a definition of a filter term of the form below. Every filter gives a list of IDs of (the last version of) those documents from a model, which satisfy conditions specified by the filter:
{\footnotesize 
   $$ GetDocsByFilter_{name}(fName, model, params) = \head(\Rec[\nil, {\tt selection}, \rev(model)])$$
}
where ${\tt selection}$ is a term of the form
{\footnotesize $$\Cond[\head(b)\in g(\alpha), \ g(\alpha)] \ [{\tt filter}(params, b), \  \cons(g(\alpha), \head(b))] \ [g(\alpha)]$$}
${\tt filter}(params, doc)$ is a formula, which represents conditions on the documents to be selected:
{\footnotesize 
$$filter(params, doc) \equiv  \text{Form}(document) = fName \wedge \varphi$$
}
where $\varphi$ is a Boolean combination of formulas of the form $s\propto t$, where $\propto\in\{<,=\}$ and $s,t$ are document terms over variables $params$, $doc$ such that in every term $GetLastDocID(m)$, $GetDocByID(x,m)$, or $GetFieldValue$ $(x,y,m)$ from $s$ or $t$, we have $m=model$.

\medskip

Next, we define the theory $\bt_s$. First of all, it contains axioms that describe fields and cardinalities for their values:
{\footnotesize  $$Field(x)  \equiv \bigvee_{f\in FieldNames} ( \ \head(x) =  f \wwedge \text{Card}(\tail(x)) \ )$$}
where $\text{Card}(y)$ is a cardinality predicate, which restricts the number of elements in a list $y$. 
We consider the following cardinalities: the list is empty; it contains zero or one element (we use notation`` ?''  for this predicate); it contains exactly one element (notation ``!'');
it contains one or more elements.
For example, ``?'' is defined as
{\footnotesize $$?(x) \equiv \forall t\in x \ \cons(\nil, t)=x$$}
The other predicates are defined similarly. 

\smallskip

Further, $\bt_s$ introduces document forms by describing which fields (with their default values) are present in a blank document of a given form:
{\footnotesize 
\begin{multline}\label{Eq:BlankDocDefinition}
    Blank(name) = document \equiv \\
    (\bigwedge_{f\in FormNames} name\neq f \  \wedge \ document = \error) \ \vee 
    \bigvee_{f\in FormNames} (name =f \ \wedge \ \varphi_{f})
\end{multline}
}
\smallskip
where $\varphi_{f}\equiv document=\nil$ or $\varphi_{f}$ has the following form, for a non-empty subset  $N_f\subseteq FieldNames$ (we assume that the elements of $N_f$ are enumerated, $N_f=\{1\lldots n\}$):
{\footnotesize 
\begin{multline*}
\exists x_1\!\in\!document \ldots \exists x_{n}\!\in\!document \\ 
\bigwedge_{i\in N_f}  \  (\head(x_i)=i \wwedge \tail(x_i)=defvalue_i \wwedge Field(x_i))  \ \wedge \forall x\in document \  (\bigvee_{i\in N_f} x=x_i)
\end{multline*}
}
where $defvalue_i$ is a list, which respects the cardinality restriction given in the definition of the $Field(x)$ predicate for $\head(x)=i$. 

The definition of the theory $\bt_s$ is complete. 

\medskip
Now we are ready to define the theory $\bt_d$. It contains definitions of daemons and a definition of a recursive $Update$ function, which given a queue, updates a model to a new state based on the definition of daemons. First, we define the $Update$ function. For the sake of readability, we split its definition into three formulas combined with disjunction and comment on them separately. 

First of all, if the queue is not empty and the first instruction in the queue is not a valid one (i.e., it is neither $CreateDoc$, $SetField$ instruction, nor a transaction name $t\in TransNames$ ) the whole queue is skipped and the model given by the $Update$ function is the initial model. If the queue is empty, then it is assumed that all the instructions in the queue have been processed and thus, $Update$ returns the value of $model$:
{\footnotesize 
\begin{equation}\label{Eq:Update-Init}
\begin{aligned}
Update(initialmodel, model, queue) = model'  \equiv \\
		( \ \head(\head(queue))\not\in\langle CreateDoc, SetField, {\tt tname_1}\lldots {\tt tname_k} \rangle \ &\wedge \\
		 queue\neq\nil \ \wedge model' = initialmodel  \ &)  \\
        \vee \ \ 
        (\ queue=\nil \wedge model' = model \ &) \\
       &\vee
\end{aligned}
\end{equation}
}
where $\{{\tt tname_1}\lldots {\tt tname_k}\}=TransNames$, for $k\geqslant 0$.

Otherwise the queue contains an instruction to create a document of a specific form, change a field value in a document having a certain ID, or launch a specific transaction. In the first case, a blank document of a given form is created (which is implemented by using existential quantification) and added to the model, the instruction is removed from the queue, and the $Update$ function is evaluated recursively on the resulting input. If a blank document of a form  with name $formName$ can not be created (due to $formName\not\in FormNames$) then the queue is skipped and $Update$ returns the initial model:
{\footnotesize 
\begin{multline}\label{Eq:Update-CreateDoc}
    ( \ \head(\head(queue)) = CreateDoc \ \wedge \ 
    \exists document  \ \  document = Blank({\tt formName}) \ \wedge \\
    ( \ (document = \error \wedge model' = initialmodel) \ \vee 
    (document\neq \error \ \wedge \\ 
    model'\!=\!Update(initialmodel, \cons(model, {\tt newdoc}), \tail(queue) ) ) \ ) \      \vee
\end{multline}
}
\noindent
where ${\tt formName}$ stands for $\head(\tail(\head(queue)))$, ${\tt newdoc}$ is a list term of the form {\footnotesize $$\lst{{\tt newsituation}, {\tt formName}, document,  \cons(GetLastDocID(model), \nil) }$$} ${\tt newsituation} = \cons( {\tt Situation}(model), \lst{{\tt formName},CreateDoc} )$, and\newline ${\tt Situation}(model)=\head(\tail(\tail(\tail(\head(model)))))$.

The case of $SetField$ instruction in the queue is formulated similarly, but the formalization is technically more complex, since modifying an already existing document requires more steps than creating a fresh one:
{\footnotesize 
\begin{multline}\label{Eq:Update-SetField}
    ( \ \head(\head(queue)) = SetField \ \wedge
    ( \ ( \ {\tt pdocument} = \error \ \vee \\ 
    \neg Field(\cons({\tt newFldValue}, {\tt fldName})) \ )  \wedge model' = initialmodel )  \ \vee \\ 
    ( {\tt pdocument} \neq \error \wedge  
    model' =  Update( \ initialmodel, \\
    \cons(model, \lst{{\tt newsituation}, {\tt form}, {\tt updateddoc}, docID}), {\tt extendedQueue})) \ ) \\
    \vee
\end{multline}
}
where ${\tt form}=\text{Form}(GetDocByID({\tt docID},model))$, ${\tt pdocument}$ denotes \phantom{aaa} {\footnotesize $FindFieldPosition(\head(\tail(GetDocByID({\tt docID},model)), model), {\tt fldName})$}\newline and ${\tt updatedDoc}$ is a shortcut for 
{\footnotesize 
$$\conc(\tail(\tail({\tt pdocument})), \cons(\head({\tt pdocument}),\cons({\tt newFldValue}, {\tt fldName})))$$
}
in which
{\footnotesize  
\begin{trivlist}
\item ${\tt docID} = \head(\tail(\head(queue)))$
\item ${\tt fldName} = \head(\tail(\tail(\head(queue))))$
\item ${\tt newFldValue} = \head(\tail(\tail(\tail(\head(queue)))))$
\item ${\tt newsituation} = \cons({\tt Situation}(model), \lst{{\tt newFldValue}, {\tt fldName}, {\tt docID}, SetField})$
\item ${\tt Situation}(model)=\head(\tail(\tail(\tail(\head(model)))))$
\end{trivlist}
}
\noindent (recall the instruction modeling conventions). 
\smallskip
Finally, ${\tt extendedQueue}$ is a shortcut for 
{\footnotesize 
\begin{equation*}
\begin{aligned}
SetFieldTrigger( docID, \ fieldName, \ newFieldValue, \ \tail(queue), \ model)
\end{aligned}
\end{equation*}
}
Thus, ${\tt updatedDoc}$ is a document with an updated field value and\newline ${\tt extendedQueue}$  is a sequence of instructions provided by a trigger on a field value change. 
By the definition above, the whole queue is skipped whenever there is no field with the specified name in a given document. Note that in this case $\tail({\tt pdocument}) = \error$ holds by the definition of $FindFieldPosition$ term.

Finally, if $\head(\head(queue))$ is a transaction name, a call to the daemon is made, which defines the corresponding transaction:
{\footnotesize 
\begin{equation}
\begin{aligned}
\bigvee_{tName\in TransNames} ( \  \head(\head(queue)) = tName \wedge model' =  \\
Update(initialmodel, model, \\
 ExecTrans(tName, {\tt docID}, {\tt params}, \tail(queue), model))) \ )
\end{aligned}
\end{equation}
}
where ${\tt docID}= \head(\tail(\head(queue)))$ is a document, for which the transaction is to be executed, and ${\tt params}=\head(\tail(\tail(\head(queue))))$ specifies parameters for the transaction.



\smallskip

Now we are in the position to define functions, which implement daemons. Their purpose is to extend the queue with a sequence of instructions depending on whether a field value in an existing document is changed or a transaction is fired. 
Both functions have similar definitions:

{\footnotesize 
\begin{equation*}
\begin{aligned}
    SetFieldTrigger(docID, fName, fValue, queue, model) & \equiv \Phi \\
    ExecTrans(tName, docID, params, queue, model) & \equiv \Psi
\end{aligned}
\end{equation*}
}
\noindent
where
{\footnotesize $$\Phi = \Cond[\theta_1, q_1]\lldots [\theta_{n}, q_{n}][queue]$$}
and for all $i\in\{1\lldots n\}$, $n\geqslant 0$, $\theta_i$ is a condition of the form
{\footnotesize $$\text{Form}(GetDocByID(docID, model))={\tt formName} \wedge fName={\tt fieldName} \wedge \varphi$$ }
(in this case $\theta_i$ is called \textbf{(${\tt formName}, {\tt fieldName}$)-\textit{condition}}) such that ${\tt formName}$ $\in FormNames$, ${\tt fieldName}\in FieldNames$ and $q_i=\conc(queue,$ ${\tt instr_i})$, where
$\varphi$ is a Boolean combination of formulas of the form ${\tt val_1}\!\propto\!{\tt val_2}$, where $\propto\in\{<, =\}$, and ${\tt val_1}, {\tt val_2}$ are document terms over variables $docID, fValue,$ $model$
and ${\tt instr_i}$ (called \textit{\textbf{queue extension}}) is a list 
{\footnotesize 
\begin{equation*}\label{Eq:QueueExtension}
\conc(s_1, \conc(s_2\ldots \conc(s_{k-1}, s_k)\ldots )
\end{equation*}
}
such that 
\medskip
\begin{itemize}
\item each $s_i$ for $i=1\lldots k$, $k\geqslant 1$ (called \textbf{\textit{instruction term}}) is a list term of the form $\lst{\lst{{\tt val},$  ${\tt fieldName'},$ $docID, SetField}}$ or $\Rec[\nil, h,$  ${\tt DocFilter}]$ with the definition: $g(\nil) = \nil, \ g(\cons(\alpha, id))=h(id)$,\newline where $h(id)=\conc(\lst{\lst{{\tt params}, id, {\tt transName}}}, \ g(\alpha))$,\newline for all $\alpha, id$ such that $cons(\alpha,$ $id)\sqsubseteq {\tt DocFilter}$ \smallskip
\item ${\tt params}$ is a list of the form $\lst{t_1\lldots t_m}$, for $m\geqslant 0$, where every $t_i$ is a document term over variables $docID, fValue, model$ \smallskip
\item ${\tt fieldName'}\in FieldNames$ and ${\tt val}$ is a document term over variables $docID, fValue, model$\newline 
\phantom{aaaaaa} (then $s_i$ is called \textbf{(${\tt formName}, {\tt fieldName'}$)-\textit{instruction}}) \smallskip
\item ${\tt DocFilter}=GetDocsByFilter_{{\tt name}}({\tt frmName}, model, {\tt p})$ \smallskip
\item ${\tt p}$ is a document term over variables $docID, fValue, model$ \smallskip
\item ${\tt name}\in FilterNames$, ${\tt frmName}\in FormNames$, and ${\tt transName}\in$\newline $TransNames$ (then $s_i$ is called \textbf{(${\tt frmName}, {\tt transName}$)-\textit{instruction}}) 
\end{itemize}
\smallskip
Thus, changing a field value in a document may cause addition of instructions to the queue, which change other fields in the same document or execute transactions over sets of documents defined by filters. 
\medskip

The formula $\Psi$ is defined similarly, but with the following minor modification (we use the notations above): 
\begin{itemize}
\item every condition $\theta_i$ has the form {\footnotesize $$\text{Form}(GetDocByID(docID,model))={\tt formName} \wedge tName={\tt transName} \wedge \varphi$$} (in this case $\theta_i$ is called \textbf{(${\tt formName}, {\tt transName}$)-\textit{condition}}), where\newline ${\tt trans\-Name}\in TransNames$ \smallskip
\item every $s_i$ is a list term of the form $\langle \lst{{\tt val}, {\tt fieldName'}, docID, SetField}\rangle$ or $\Rec[\nil, h, {\tt DocFilter}]$, where $h$ is given as $\conc(\lst{\lst{{\tt params'}, id,$ ${\tt transName}}},$ $\ g(\alpha))$ or $\conc(\lst{\lst{{\tt frmName},$ $CreateDoc}}, \ g(\alpha))$, or $s_i$ is of the form $\lst{\lst{{\tt frmName},$ $CreateDoc}}$ (in the latter two cases $s_i$ is called \textbf{(${\tt frmName},$ $CreateDoc$)-\textit{instruction}}) \smallskip
\item ${\tt val}$,  ${\tt val_1},$ ${\tt val_2}$, ${\tt p}$ are document terms over variables $docID, params,$ $model$ and ${\tt params'}$ is a list of the form $\lst{t_1\lldots t_m}$, for $m\geqslant 0$, where every $t_i$ is a document term over variables $params, model$.
\end{itemize}
\smallskip
Thus, executing a transaction over a document may cause addition of instructions to the queue, which change fields in the document, create new documents (of the same or different document form), or execute transactions over sets of documents defined by filters.
\smallskip

The definition of the document theory $\bt$ is complete.

\smallskip

Let the \textit{size} of $\bt$ be the total size of its axioms (given as strings).

\section{Termination of Transactions}

The recursive definition of $Update$ function yields the natural notion of \textit{chase} operator, which for a given document theory $\bt$ and constant lists $model, queue$, where $queue\neq\nil$, outputs lists $model'$ and $queue'$ obtained after processing the first instruction from $queue$ (i.e., $\head(queue))$. In other words, for any list superstructure $HW(\mathcal{M})$, it holds {\footnotesize $$HW(\mathcal{M}) \models Update(model, model, queue)=Update(model, model', queue')$$} where $model'$ is obtained from $model$ by the definition of $Update$ function in $\bt$ without applying recursion and either $queue'$ is obtained in the same way from $queue$, or it holds that $queue=\nil$. We denote this fact as $\lst{model,queue}\mapsto\lst{model',queue'}$. A \textit{chase sequence} wrt $\bt$ for a list $\lst{m_0,q_0}$ of the form above is a sequence of lists $\lst{m_0,q_0}, \lst{m_1,q_1}, \ldots$, where $\lst{m_i,q_i}\mapsto\lst{m_{i+1},q_{i+1}}$, for all $i\geqslant 0$. A chase sequence is \textit{terminating} if it is of the form $s_0\lldots s_k$, for some $k\geqslant 1$, where $s_k=\lst{m_k,\nil}$. 

In the following, we note that there may not exist a terminating chase sequence for a given list $\lst{m,q}$ and a theory $\bt$. Then we formulate a sufficient condition on the form of $\bt$, which guarantees chase termination, and finally we estimate the complexity of computing the chase.

\begin{theorem}[Termination of Transactions is Undecidable]\label{Thm:Termination}
It is undecidable whether there is a terminating chase sequence for a list $s=\lst{model,$ $queue}$ wrt a document theory $\bt$.
\end{theorem}
\begin{proof}
The theorem is proved by a reduction of the halting problem for Turing machines. We define a Turing machine (TM) as a tuple $M=(Q,\mathcal{A},\delta)$, where $Q$ is a set of states, $\mathcal{A}$ an alphabet containing a distinguished blank symbol $\b$, and $\delta: Q\times\mathcal{A} \mapsto  Q\times\mathcal{A}\times\{-1,1\}$ a (partial) transition function. We assume w.l.o.g. that the tape of $M$ is right-infinite and a configuration of $M$ is a word over $Q\cup\mathcal{A}$, which contains exactly one state symbol $q\in Q$. An initial configuration is a word $\c_0$ of the form $\b\qo\b\ldots \b$, where $\qo\in Q$. For a configuration $\c$, the successor configuration is defined by $\delta$ in a standard way and is denoted as $S_\delta(\c)$. We define the halting problem as the set of TMs, for which there is a finite sequence of configurations $\c_0,\ldots , \c_k$, $k\geqslant 0$ such that $\c_{i+1}=S_\delta(\c_i)$, for all $0 \leqslant i < k$, and $S_\delta(\c_{k})$ is undefined.
Given a TM $M$, we define a document model theory $\bt$, which encodes $M$. 

Let $FormNames=\{TMcell\}$, $FieldNames=\{TMsymbol\}$, $TransNames$ $=\{MakeTMStep\}$, and $FilterNames=\varnothing$. 

The theory $\bt$ contains the following axioms, which specify the single document form and field used for representing the content of the tape of $M$:
{\footnotesize 
\begin{equation}
Field(x) \ \equiv \ \head(x) = TMsymbol \ \wedge \ !(tail(x))
\end{equation}
\begin{multline}
    Blank(name) = document \ \\ \equiv \ 
    ( name\neq TMcell \wwedge document = \error) \ \vee \ 
    \varphi_{name}
\end{multline}
}
\smallskip
where $\varphi_{name}$ has the form
{\footnotesize 
\begin{multline*}
name = TMcell \wwedge  \exists x \in document  \\ 
(\ \head(x)=TMsymbol \wwedge \tail(x) = \lst{\b} \wwedge Field(x) \wwedge \forall y\in document \  (x=y) \ ) 
\end{multline*}
}
Since $FilterNames=\varnothing$, the theory $\bt$ contains no axioms for filter functions, so  we formulate next the definition of the recursive $Update$ operator. It is given by the disjunction of formulas (\ref{Eq:Update-Init})--(\ref{Eq:Update-SetField}) with the following formula:
{\footnotesize 
\begin{multline}
\begin{aligned}
\head(\head(queue) = MakeTMStep  \ \wedge \ 
model' =  Update(\ initialmodel, model, \\
 ExecTrans(MakeTMStep, {\tt docID}, \nil, \tail(queue), model) \ )
\end{aligned}
\end{multline}
}
where ${\tt docID}= \head(\tail(\head(queue)))$. 

\smallskip

Finally, the daemons are defined as follows. The first definition is trivial, it says that changing a field value does not yield any extension of the queue:
{\footnotesize 
$$SetFieldTrigger(docID, fieldName, fieldValue, queue, model) \equiv queue$$
}
The second daemon encodes transitions of $M$:
{\footnotesize 
\begin{equation}
\begin{aligned}
ExecTrans(tName, docID, params, queue, model) \equiv \\
 \Cond[docID=1, \nil][\theta_1, q_1]\lldots [\theta_{n}, q_{n}][\nil]
\end{aligned}
\end{equation}
}
such that there is one-to-one correspondence between the set of pairs $[\theta_i, q_i]$, for $i=1\lldots n$, $n\geqslant 0$, and the graph of the transition function $\delta$ given as follows (we assume below that $a, a'\in\mathcal{A}$ and $q, q'\in Q$).

If $\delta((a,q))=(a',q',-1)$, then there is $i\in\{1\lldots n\}$ such that
{\footnotesize
\begin{multline*}
\theta_i = \text{Form}(GetDocByID(docID))=TMCell \ \wedge \ tName = MakeTMStep \ \wedge \ \\ GetFieldValue(\tail(docID),\ TMSymbol)=\lst{a} \ \wedge \ \\
GetFieldValue(docID, \ TMSymbol)=\lst{q}
\end{multline*}
}
and $q_i=\conc(queue, \lst{s_3, s_2, s_1})$, where 
\begin{itemize}
\item $s_1=\lst{\lst{q'}, TMSymbol, \ \tail(docID),  \ SetFieldValue}$
\item $s_2=\lst{\lst{a'}, TMSymbol, \ docID, \ SetFieldValue}$
\item $s_3=\lst{\nil, \ \tail(docID), \ MakeTMStep}$
\end{itemize}

\medskip

If $\delta((a,q))=(a',q',1)$, then there is $i\in\{1\lldots n\}$ such that $\theta_i$ is of the form above 
and
{\footnotesize $$q_i=\conc(queue, \ \Cond[docID=GetLastDocID(model), \ {\tt add\& updtape}][{\tt updtape}]),$$} 
where 
\begin{itemize}
\item ${\tt add\& updtape}=\cons({\tt updatetape'}, \ \lst{TMCell, CreateDoc})$
\item ${\tt updatetape'} = \lst{s_4, s_3, s_2, s_1}$
\item $s_1=\lst{\lst{a'}, \ TMSymbol, \ \tail(docID), \ SetFieldValue}$
\item $s_2=\lst{\lst{\b}, \ TMSymbol, \ docID, \ SetFieldValue}$
\item $s_3=\lst{\lst{q'}, \ TMSymbol, \ \cons(docID,\nil), \ SetFieldValue}$
\item $s_4=\lst{\nil, \ \cons(docID,\nil), \ MakeTMStep}$
\item ${\tt updatetape} = \lst{s_4, s_3, s_2', s_1}$
\end{itemize}
and
{\footnotesize 
\begin{multline*}
s_2'=\lst{GetFieldValue(\cons(docID,\nil), TMSymbol, model), \\
TMSymbol, \ docID, \ SetFieldValue}
\end{multline*}
}

The definition of document theory $\bt$ is complete. 

\medskip

Now we define a list of instructions $initqueue$, which encodes the first two symbols of the initial configuration of $M$ and enforces execution of $MakeTMStep$ transaction over a document, in which the value of the field $TMSymbol$ is $\qo$. We set {\footnotesize $$initqueue=\conc(\lst{{\tt runTM}}, {\tt filltape})$$} where ${\tt filltape}=\lst{ \ \lst{\lst{\qo}, TMSymbol, 2, SetFieldValue},  \ \lst{TMcell,$ $CreateDoc},$ $\lst{TMcell, CreateDoc} \ }$ and ${\tt runTM}=\lst{\nil, 2, MakeTMStep}$. 

It can be shown that there is a terminating chase sequence for $\lst{\nil, initqueue}$ iff $M$ halts. \QED
\end{proof}

A close inspection of the theory $\bt$ from Theorem \ref{Thm:Termination} shows that non-termination may be caused by the ability to change a field value of the same document or execute the same transaction infinitely many times. Thus, in general the definition of $SetFieldTrigger$ and $ExecTrans$ functions of $\bt$ allows for cyclic references between instructions and transactions. In the following, we demonstrate that if one forbids cycles then chase termination is guaranteed.

\begin{definition}[Dependency Graph]
A \textit{dependency graph} over a document theory $\bt$ is a directed graph with the set of vertices $V$ equal to $FormNames\times (FieldNames\cup TransNames\cup\{CreateDoc\})$ and the set of edges $E$ defined as follows. 

\noindent For any $(\!form,\!name), (\!form',\! name'\!)\!\in\!V\!$, there is an edge from $(\!form,\! name\!)$ to $(\!form',\!name'\!)$ if there is $[\theta,q]$ in the definition of $SetFieldTrigger$ or $ExecTrans$ functions in $\bt$, in which $\theta$ is a $(form, name)$-condition and $q=\conc({\tt queue},$ ${\tt instr})$, for a list ${\tt queue}$ and queue extension ${\tt instr}$, such that there is a $(form',$ $name')$-instruction in the definition of ${\tt instr}$.
\end{definition}

\begin{definition}[Locally Simple Document Theory]
A document theory $\bt$ is called \textit{locally simple} if the dependency graph over $\bt$ is acyclic.
\end{definition}

We now introduce several auxiliary notions, which will be used in two theorems below. 
We slightly abuse our terminology and for a document theory $\bt$ we call a constant list $t$ \textit{instruction} if it is of the form $\lst{formName,$ $CreateDoc}$ (a \textit{CreateDoc-instruction}), or $\lst{val, fieldName, docID,$ $SetField}$ (\textit{SetField-instruction}), or $\lst{params, docID, transName}$ (\textit{transaction}), where $formName$ $\in FormNames$, $fieldName\in FieldNames$, $transName\in TransNames$, and $val, params, docID$ are lists such that $docID=\bar{n}$, for some $n\in \omega$. 

\smallskip

Let $G$ be a dependency graph over $\bt$, and $model$ a list. An instruction $t$ is said to have \textit{rank} $k$ wrt $\bt, model$, for $k\geqslant 0$, if
\begin{itemize}
\item $t$ is a $SetField$-instruction or a transaction as above, respectively, and there is an element of the form $\lst{sit, formName, doc, docID}\in model$, for lists $sit, doc$ and $formName\in FormNames$, such that the longest path outgoing from $(formName, fieldName)$ or $(formName, transName)$ in $G$, respectively, has $k+1$ vertices; \smallskip
\item the above conditions do not apply and $k=0$.
\end{itemize} 

\begin{theorem}[Local Simplicity Implies Termination of Transactions]\label{Thm:LocalSimplicityGivesTermination}
For any locally simple document theory $\bt$ and constant lists ${\tt model}$, ${\tt queue}$, there is a terminating chase sequence for $\lst{{\tt model}, {\tt queue}}$ wrt $\bt$. 
\end{theorem}
\begin{proof}
We show that for any such ${\tt model}$ and ${\tt queue}$, there is a finite chase sequence $s_0,\ldots , s_n$, where $s_0=\lst{{\tt model}, {\tt queue}}$, $n\geqslant 1$, such that $s_n=\lst{model', \tail({\tt queue})}$, where $|model'|=|{\tt model}|+p$, for some $p\geqslant 0$. This yields that any instruction from ${\tt queue}$ can be processed in a finite number of steps, from which the claim follows. 

Let $s_0,s_1\ldots $ be a chase sequence for $s_0$ and let $s_1=\lst{m, q}$. We use induction on the rank $k$ of instruction $t = \head({\tt queue})$ wrt $\bt, {\tt model}$.

If $k=0$ then there are two possible cases:
\begin{itemize}
\item $t$ is a SetField-instruction, $|m|=|model|$, and $q=\nil$ or $q=\tail({\tt queue})$ (since $\bt$ is locally simple, no new instructions can appear in $q$ after processing a SetField-instruction of rank $0$) \smallskip
\item $t$ is a CreateDoc-instruction, $|m|\leqslant |model|+1$, and $q=\nil$ or $q=\tail({\tt queue})$ (since by the definition of the $Update$ function from a document theory, no new instructions can appear in $q$ after processing a CreateDoc-instruction)
\end{itemize}
If $k\geqslant 1$ then $t$ is not a CreateDoc-instruction, hence, $|m|=|{\tt model}|$ and either $q=\nil$ or $q=\tail({\tt queue})$, or $q=\conc(\tail({\tt queue}), c)$, where $c=\lst{t_1\lldots t_j}$, $j\geqslant 1$, is a list of instructions of rank smaller than $k$. Then by applying the induction assumption we obtain the required statement. \QED
\end{proof}

Although the theorem states that local simplicity guarantees termination, it does not provide any insight on how difficult it is to compute the effects of transactions. The next result indicates that the complexity is high, which is due to the possibility to create exponentially many documents by using recursive instruction terms. 

\medskip

For $n\geqslant 0$, let $\EXP{1}{n}$ be the notation for $2^n$ and for $k\geqslant 1$, let $\EXP{(k+1)}{n}=2^\EXP{k}{n}$.

\begin{theorem}[Computing Effects of Transactions is Hard]\label{Thm:Hardness}
For any $k\geqslant 1$, $n\geqslant 0$, there exists a locally simple document theory $\bt$ and a constant list ${\tt queue}$, both of sizes linear in $k,n$, such that the terminating chase sequence for $s_0=\lst{\nil, {\tt queue}}$ wrt $\bt$ has the form $s_0\lldots s_m$, where $m\geqslant\EXP{k}{n}$ and $s_m=\lst{model,\nil}$, for a list $model$ such that $|model|\geqslant\EXP{k}{n}$.
\end{theorem}
\begin{proof}
Given numbers $k,n$, let us define a locally simple document theory $\bt$ as follows. 
Let $FormNames=\{Form_0\lldots Form_k\}$, $FieldNames=\varnothing$,\newline $FilterNames\!\!=\!\!\{SelectAllbyForm\}$, and $TransName\!\!=\!\!\{Duplicate,$  $Duplicate'\}$ $\cup\{MakeExp_1,$ $\ldots , MakeExp_k\}$. There are no definitions of document fields in the theory $\bt$, since $FieldNames=\varnothing$, and hence, the definition of document forms given by equation (\ref{Eq:BlankDocDefinition}) is trivial.


The important part is the definition of document filters and daemons. For every $i=1\lldots k$, there is a definition of a filter function, which gives id's of all documents of the form $fName$: 
{\footnotesize 
\begin{equation*}
\begin{aligned}
    GetDocsByFilter_{SelectAllbyForm}(fName, model, params) = \\
    \head(\Rec[\nil, \ {\tt selection}, \ \text{rev}(model)])
\end{aligned}    
\end{equation*}
}
where ${\tt selection}$ is a term of the form
{\footnotesize 
\begin{equation*}
\Cond[\head(b)\in g(\alpha), \ g(\alpha)] \ [Form(b)=fName, \  \cons(g(\alpha), head(b))]\ [g(\alpha)]   
\end{equation*}
}

The functions, which implement daemons, are defined as follows. The first one is trivial:
{\footnotesize $$SetFieldTrigger(docID, fName, fValue, queue, model) \equiv queue$$
}
The second one defines transactions, which duplicate the number of documents in a model:
{\footnotesize
\begin{multline*}
ExecTrans(tName, docID, params, queue, model) \equiv \\
 \Cond [\phi_1, p_1]\ldots [\phi_{k}, p_{k}] [\psi_1, q_1]\ldots [\psi_{k}, q_{k}][\nil]
\end{multline*}
}
such that for $i=1\lldots k$: 

$\phi_i \ = \ \text{Form}(GetDocByID(docID))=Form_i \ \wedge \ tName = MakeExp_i,$

$p_i=\conc(queue, \ \Rec[\nil, h, {\tt DocFilter}])$, where
\begin{center}
{\footnotesize $h \ = \ \conc( \ \lst{\lst{\nil, \ id, \ Duplicate}}, \ g(\alpha) \ )$} 
{\footnotesize ${\tt DocFilter}\ =\ GetDocsByFilter_{SelectAllbyForm}(Form_{i-1}, \ model, \ \nil)$}
\end{center}

\medskip

$\psi_i = Form(GetDocByID(docID))=Form_{i-1} \ \wedge \ tName = Duplicate$, 

$q_i=\conc(queue, \Rec[\nil, h, {\tt DocFilter}])$, where
\begin{center}
{\footnotesize $h \ = \ \conc( \ \lst{\lst{Form_{i}, \ CreateDoc}}, \ g(\alpha) \ )$} 
{\footnotesize ${\tt DocFilter}\ =\ GetDocsByFilter_{SelectAllbyForm}(Form_{i}, \ model, \ \nil)$}
\end{center}




The definition of the theory $\bt$ is complete. It can be readily verified that $\bt$ is locally simple: the non-singleton connected components of the dependency graph over $\bt$ are given  by pairs {\footnotesize $$(\lst{Form_i, MakeExp_i}, \ \lst{Form_{i-1}, Duplicate})$$} for all $i=1\lldots k$.

\medskip

Now let us define a list ${\tt queue} = \conc({\tt run}, \ {\tt init})$, where 
{\footnotesize 
\begin{multline*}
{\tt run}=\lst{\lst{\nil, \ GetLastDocID(model), \ MakeExp_k}\lldots \\
 \lst{\nil, \ GetLastDocID(model), \ MakeExp_1}}
\end{multline*}
}
and ${\tt init}$ is a list of the form
{\footnotesize
\begin{multline*}
\lst{\lst{Form_k,CreateDoc}\lldots\lst{Form_1,CreateDoc}, \\ \underbrace{\lst{Form_0,CreateDoc}\lldots\lst{Form_0,CreateDoc}}_{n \ \text{times}}}
\end{multline*}
}

\smallskip

It can be shown by induction on the number $k$ that there is a terminating chase sequence $s_0\lldots s_m$ for $s_0=\lst{\nil, {\tt queue}}$ such that $s_m=\lst{model, \nil}$,  {\footnotesize $$|model|=n+\sum_{i=1\lldots k}\EXP{i}{n}$$} (there are $n$ documents of the form $Form_0$ and $\EXP{i}{n}$-many documents of the form $Form_i$ in $model$, for all $i=1\lldots k$), and {\footnotesize $$m=(n+2k)+n+\sum_{i=1\lldots k}(\EXP{i}{n}-1) + \sum_{i=1\lldots k-1}\EXP{i}{n}$$} 

where $n+2k$ is the number of instructions in ${\tt queue}$, the additional $n$ is the number of $\lst{\!\!Form_0,  Duplicate\!\!}$ transactions (generated by a $\lst{\!\!Form_1, MakeExp_1\!\!}$ transaction), and for each $i=1\lldots k$, there are $(\EXP{i}{n}-1)$-many $\lst{Form_{i},$ $CreateDoc}$ instructions (generated by $\lst{Form_{i-1}, Duplicate}$ transactions) and for each $i=1\lldots k-1$, there are $\EXP{i}{n}$-many $\lst{Form_{i},$ $Duplicate}$ transactions (generated by a $\lst{Form_{i+1},$ $MakeExp_{i+1}}$ transaction). \QED 
\end{proof}

Finally, let us formulate a sufficient condition, which guarantees polynomial boundedness of effects of transactions. 
Let $G$ be a dependency graph over a document theory $\bt$ and for $form\in FormNames$, $name\in FieldNames\cup TransNames\cup\{CreateDoc\}$, let $s$ be a $(form,name)$-instruction in a queue extension from the definition of $SetFieldTrigger$ or $ExecTrans$ functions in $\bt$. We call the term $s$ \textit{document generating} if either $name=CreateDoc$ or  $s=\Rec[\nil, h, {\tt DocFilter}]$, where $h=\conc(\lst{\lst{form, CreateDoc}}, \ g(\alpha))$, or $(form,name)$ has a successor vertex $(form',name')$ in $G$, which is given by a document generating term. 

\begin{theorem}[Polynomially Bounded Effects of Transactions]\label{Thm:Polytime}
Let $\bt$ be a locally simple document theory such that in any queue extension from the definition of $SetFieldTrigger$ or $ExecTrans$ functions in $\bt$, there are no document generating \Rec-terms.

Then for any constant list ${\tt model}$ and a list of instructions ${\tt queue}$, the terminating chase sequence for $s_0=\lst{{\tt model}, {\tt queue}}$ has the form $s_0\lldots s_n$, where $n$ is exponentially bounded by the size of $\bt$ and $s_0$, and $s_n=\lst{model', \nil}$, for a list $model'$ of size polynomially bounded by the size of $\bt$ and $s_0$.
\end{theorem}
\begin{proof}
Let $N$ be the maximal number of instruction terms in a queue extension from the definition of $SetFieldTrigger$ or $ExecTrans$ functions in $\bt$. Clearly, $N$ is bounded by the size of $\bt$.

Let ${\tt model}$, ${\tt queue}$ be lists, which satisfy the conditions of the lemma, and let $k$ be the rank of instruction $t = \head({\tt queue})$ wrt $\bt, {\tt model}$. By definition, $k$ is bounded by the number of vertices in the dependency graph over $\bt$ and thus, it is bounded by the size of $\bt$. 
We show that there is a chase sequence $s_0, s_1,\ldots , s_n$ such that $s_0=\lst{{\tt model}, {\tt queue}}$, $s_n=\lst{model', \tail({\tt queue})}$, $|model'|\leqslant |{\tt model}|+N$, and $n\leqslant (N\cdot (|model|+N))^k$. Then it follows that there is a terminating chase sequence $s_0\lldots s_m$ for $s_0$, where $s_m=\lst{m,\nil}$, $|m|\leqslant |{\tt queue}|\cdot (|{\tt model}|+N)$,  and $m\leqslant |{\tt queue}|\cdot (N\cdot (|model|+N))^k$, which proves the theorem.

We use induction on $k$. Let $s_0,s_1\ldots $ be a chase sequence for $s_0$ and let $s_1=\lst{m, q}$. 

The case $k=0$ is treated like in the proof of Theorem \ref{Thm:Termination}. If $k\geqslant 1$ then $t$ is not a CreateDoc-instruction, hence, $|m|=|{\tt model}|$ and either $q=\nil$ or $q=\tail({\tt queue})$, or $q=\conc(\tail({\tt queue}), c)$, where $c=\lst{t_1\lldots$ $t_j}$, $j\geqslant 1$, is a list of instructions of rank smaller than $k$. If $t$ is a SetField-instruction then it follows from the definition of $SetFieldTrigger$ and filter functions that $j\leqslant N\cdot |{\tt model}|$ (and hence, $j\leqslant N\cdot (|{\tt model}| + N)$) and by the condition of the lemma, there is no document generating instruction in $c$. Then the induction assumption gives the required statement.

If $t$ is a transaction then for some $n_c, n_r$ such that $n_c+n_r=N$, $c$ contains at most $n_c$-many CreateDoc-instructions and in total at most $n_r\cdot |{\tt model}|$-many SetField-instructions or transactions. 

Moreover, by the definition of $SetFieldTrigger$ function and the condition of the lemma, the latter are not document generating. Hence, for any chase sequence $s_0, s_1,\ldots$ and any $s_i=\lst{m', q'}$, where $i\geqslant 1$, it holds that $|m'|\leqslant |{\tt model}|+n_c \leqslant |{\tt model}|+N$. 

Since $n_c, n_r\leqslant N$, the number of instructions in $c$ is bounded by $N\cdot |{\tt model}| +N$, which is anyway less or equal than $N\cdot (|{\tt model}| + N)$, and each of these instructions is of rank smaller than $k$. Then by the induction assumption we obtain the required statement. \QED
\end{proof}

\section{Conclusions}
We have shown that document theories (and thus, the Document Modeling approach) implement a Turing-complete computation model even in the presence of a tractable language of arithmetic operations (over document field values) and queries (for selecting collections of documents). This confirms that one of the main sources of the computational complexity are the definitions of daemons, which specify transactions and relationships between them. If the definitions are given in a way that allows for executing the same transaction or changing the value of a document field infinitely many times, then it is possible to implement computations of any Turing machine. We have shown that disallowing cyclic relationships between transactions guarantees decidability of transaction termination (importantly, cycles can be easily detected by a syntactic analysis of axioms of a document theory), but the complexity of computing effects of transactions even in this case is high, if creating documents in loops is possible. In fact, using looping in transactions is natural, since it allows for performing updates over collections of documents. If documents can be only modified in loops, but not created, then the complexity of computing effects of transactions is decreased and we have noted a case when the effects are polynomially bounded. In further research, we plan to make a more detailed complexity analysis for various (practical) restrictions on the definition of daemons. In this paper, we did not study the contribution of query languages to the complexity of computing effects of transactions and we have adopted a relatively simple query language. Since daemons employ document queries to modify collections of documents, it would be important to study the interplay between these two sources of complexity.


\bigskip

\begin{thebibliography}{999}


\bibitem{bib41}
Ershov Yu.L., Goncharov S.S., Sviridenko D.I.
Semantic Programming. {\it Information processing 86: Proc. IFIP 10th
World Comput. Congress}. 1986, vol. 10, Elsevier Sci., Dublin, 
pp. 1093--1100.

\bibitem{bib51}
Ershov Yu.L., Goncharov S.S., Sviridenko D.I.
Semantic Foundations of Programming. {\it Fundamentals of Computation
Theory: Proc. Intern. Conf}. FCT 87, Kazan, 116--122. Lect.
Notes Comp. Sci., 1987, vol. 278. 

\bibitem{bib52} 
Goncharov S.S., Sviridenko D.I. $\Sigma$-programming.
Transl. II. {\it  Amer. Math. Soc.}, 1989, no. 142, pp. 101-121.

\bibitem{bib22}
Goncharov S.S., Sviridenko D.I. $\Sigma$-programming
and its Semantics. {\it Vychisl. Systemy}, 1987, no. 120, pp. 24-51. (in
Russian).

\bibitem{bib32}
Goncharov S.S., Sviridenko D.I. Theoretical Aspects of
$\Sigma$-programming. {\it Lect. Notes Comp. Sci.}, 1986, vol. 215, pp. 169-179. 

\bibitem{condterms}
Goncharov S.S. Conditional Terms in Semantic Programming.
{\it Siberian Mathematical Journal}, 2017, vol. 58, no. 5, pp. 794-800. 

\bibitem{Doklady}
Goncharov S.S., Sviridenko D.I.
The Logic Language of Polynomial Computability. 
{\it Doklady Mathematics}, 2019, vol. 99, no.2, pp. 11–14. 

\bibitem{Recursive}
Goncharov S.S., Sviridenko D.I. Recursive Terms in Semantic Programming. 
{\it Siberian Mathematical Journal}, 2018, vol. 59, no. 6, pp. 1279-1290.

\bibitem{docmodels2} Kazakov I.A., Kustova I.A., Lazebnikova E.N.,  Mantsivoda A.V.  Building locally simple models: theory and practice. {\it The Bulletin of Irkutsk State University. Series Mathematics}, 2017, vol. 21, pp. 71-89. (in Russian). 

\bibitem{docmodels}
Malykh A.A., Mantsivoda A.V.
Document modeling.  {\it The Bulletin of Irkutsk State University. Series Mathematics}, 2017, vol. 21, pp. 89-107. (in Russian). 

\bibitem{MantsivodaPonomaryov}
Mantsivoda A.V., Ponomaryov D.K.  
A Formalization of Document Models with Semantic Modelling. 
{\it Bulletin of Irkutsk State University, Series Mathematics}, 2019. vol. 27, pp. 36-54.

\bibitem{OspichevPonomaryov}
Ospichev S., Ponomarev D.
On the Complexity of Formulas in Semantic Programming. {\it  Siberian Electronic Mathematical Reports}, 2018, vol. 15, pp. 987-995.


\bibitem{ReiterBook2}  Reiter R. {\it Knowledge in Action: Logical Foundations for Describing and Implementing Dynamical Systems}. MIT Press, 2001. 

\bibitem{Vityaev} Vityaev E.E. Semantic Probablistic Inference of Predictions. 
{\it The Bulletin of Irkutsk State University. Series Mathematics}, 2017, vol. 21, pp. 33-50. (in Russian).

\end{thebibliography}
\end{document}